\documentclass[3p,a4paper,twocolumn]{elsarticle}
\usepackage{psfrag}
\usepackage{graphicx}
\usepackage{bbold}

\newtheorem{theorem}{Theorem}

\newtheorem{proposition}[theorem]{Proposition}

\newenvironment{proof}[1][Proof]{\noindent\textbf{#1.} }{\ \rule{0.5em}{0.5em}}

\begin{document}

\title{Nonadditive entropy for random quantum spin-S chains}

\author{A. Saguia}
\ead{amen@if.uff.br}

\author{M. S. Sarandy}
\ead{msarandy@if.uff.br}

\address{Instituto de F\'{\i}sica, Universidade Federal Fluminense,
Av. Gal. Milton Tavares de Souza s/n, Gragoat\'a, 24210-346, Niter\'oi, RJ, Brazil.}

\begin{abstract}
We investigate the scaling of Tsallis entropy in disordered quantum spin-S chains. We show that 
an extensive scaling occurs for specific values of the entropic index. Those values depend only 
on the magnitude S of the spins, being directly related with the effective central charge associated 
with the model.

\end{abstract}


\maketitle

%%%%%%%%%%%%%%%%%%%%%%%%%%%%%%%%%%%%
\section{Introduction}
%%%%%%%%%%%%%%%%%%%%%%%%%%%%%%%%%%%%

Correlations among parts of a quantum system are behind remarkable phenomena, 
such as a quantum phase transition (QPT)~\cite{Sachdev:book,Continentino:book}. 
In particular, the relationship between correlations and QPTs is revealed by 
the behavior of entanglement at criticality as measured, e.g., by the von Neumann entropy 
(see, for instance, Ref.~\cite{Amico:08}). Given a quantum system in a pure state $|\psi\rangle$ 
and a bipartition of the system into two subsystems $A$ and $B$, the von Neumann entropy between 
$A$ and $B$ reads
\begin{equation}
{\cal S}=-\textrm{Tr} \left( \rho_A \ln \rho_A \right) = -\textrm{Tr} 
\left( \rho_B \ln \rho_B \right),
\label{vonNeumann}
\end{equation}
where $\rho_A=\textrm{Tr}_B \rho$ and $\rho_B = \textrm{Tr}_A \rho$ denote the reduced 
density matrices of $A$ and $B$, respectively, with $\rho=|\psi\rangle\langle \psi|$. 
If $A$ and $B$ are probabilistic independent (such that $\rho = \rho_A \otimes \rho_B$), the von 
Neumann entropy is additive, i.e., ${\cal S}_{AB} = {\cal S}_A + {\cal S}_B$. As a consequence, 
${\cal S}$ is extensive for uncorrelated subsystems, namely, ${\cal S}(L) \propto L$, where $L$ 
denotes the size of a block of the system. On the other hand, ${\cal S}$ becomes nonextensive in presence 
of correlations. Indeed, for critical systems in one dimension, which are known to be highly entangled, 
conformal invariance implies a diverging logarithmic scaling given by ${\cal S}(L) \propto (c/3) \ln L$ 
%%% v3
(or, more specifically, ${\cal S}(L) = (c/3) \ln L + \textrm{constant}$),
%%%%%, 
where $c$ is the central charge associated with the Virasoro algebra of the underlying conformal field 
theory~\cite{Vidal:03,Korepin:04,Calabrese:04}.
For noncritical (gapful) systems in one dimension, entanglement 
saturates at a constant value $k$, i.e., ${\cal S}(L) \rightarrow k$ as $L \rightarrow \infty$. 
More generally, for higher dimensions, noncritical systems are expected to obey the area law, which 
implies that the von Neumann entropy of a region scales as the surface area of the region instead 
of the volume of the region itself. In other words, the area law establishes that 
${\cal S}(L) \propto L^{D-1}$ ($L\rightarrow \infty$), where $D$ is the dimension of the system.     

Remarkably, it has recently been shown in Refs.~ \cite{Tsallis:05,Caruso:08} that a quantum system 
may exhibit specific probability correlations among its parts such that an extensive entropy can be 
achieved even for highly correlated subsystems. This has been obtained by generalizing the von Neumann 
entropy into the nonadditive Tsallis q-entropy~ \cite{Tsallis:88,Tsallis:book}
\begin{equation}
{\cal S}_q [\rho]= \frac{1}{1-q} \left(\textrm{Tr} \rho^q - 1 \right), 
\label{qentropy}
\end{equation}
with $q \in \mathbb{R}$. One can show that the von Neumann entropy is a particular case of Eq.~(\ref{qentropy}) by taking $q=1$. 
Tsallis entropy has been successfully applied to handle a variety of physical systems, 
in particular those exhibiting long-range interactions. Recent experimental results for its predictions can 
be found, e.g., in Refs.~\cite{Douglas:06,Liu:08}. In Tsallis statistics, the additivity of the von Neumann 
entropy for independent subsystems is replaced by the pseudo-additivity relation of the ${\cal S}_q$ entropy
\begin{eqnarray}
{\cal S}_q [\rho_A \otimes \rho_B] = {\cal S}_q [\rho_A] + {\cal S}_q [\rho_B] \nonumber \\
+\left(1-q\right) {\cal S}_q [\rho_A] {\cal S}_q [\rho_B].
\label{pseudo-add}
\end{eqnarray}
The investigation of ${\cal S}_q$ in conformal invariant quantum systems has revealed that the extensivity of the entropy 
can be achieved for a particular choice $q_{ext}$ of the entropic index $q$ in Eq.~(\ref{qentropy}). In particular, 
$q_{ext}$ is directly associated with the central charge $c$. More specifically, the extensivity of ${\cal S}_q$ 
occurs for~\cite{Caruso:08}
\begin{equation}
q_{ext} = \frac{\sqrt{9+c^2} - 3}{c}. 
\label{qext}
\end{equation}
The aim of this work is to consider the scaling of the nonadditive entropy ${\cal S}_q$ and, consequently, its 
extensivity in quantum critical spin chains under the effect of disorder into the exchange couplings among the spins. 
Indeed, disorder appears as an essential feature in a number of condensed matter systems, motivating a great 
deal of theoretical and experimental research (see, e.g., Refs.~\cite{Young:98,Igloi:05}). In particular, 
it is well known that, in the case of a spin-S random exchange Heisenberg antiferromagnetic chain (REHAC), 
disorder can drive the system to the so-called random singlet phase (RSP), which is a gapless phase described   
by spin singlets distributed over arbitrary distances~\cite{Fisher:94}. In recent years, it has been observed that the 
entanglement entropy in critical random spin chains displays a logarithmic scaling that closely resembles the 
behavior of pure (non-disordered) systems. Indeed, for a block of spins of length $L$, we have that the von 
Neumann entropy reads ${\cal S}(L) \propto (c_{eff}/3) \ln L$, where $c_{eff}$ is an effective central charge 
that governs the scale of the entropy~\cite{Refael:04}. Moreover, it has been shown that in the case of the RSP, 
$c_{eff}$ is determined solely in terms of the magnitude $S$ of the spin in the chain~\cite{Saguia:07,Bonesteel:07,Refael:07} 
%%% v3
(see Ref.~\cite{Refael:09} for a review of entanglement in random systems and Ref.~\cite{others} for other connected results).
%%%
Here, we will show that the extensivity of ${\cal S}_q$ can also be obtained for random critical spin chains, with $q_{ext}$ governed by 
$c_{eff}$. Hence, $q_{ext}$ will be given as a unique function of the spin $S$. Moreover, as we will see, around the extensivity 
point $q_{ext}$, ${\cal S}_q(L) \propto L^\gamma$, with the exponent $\gamma$ of the power law given by a quadratic function of $q$. 
 
%%%%%%%%%%%%%%%%%%%%%%%%%%%%%%%%%%%%%%%%%%%%%%%%%%%%%%%%%%%%%%%%%%%%%%
\section{Nonadditive entropy for a set of random singlets}
%%%%%%%%%%%%%%%%%%%%%%%%%%%%%%%%%%%%%%%%%%%%%%%%%%%%%%%%%%%%%%%%%%%%%%%

We begin by considering the typical arrange of a quantum spin-S chain in the RSP, 
which is provided by a set of spin singlets distributed over arbitrary distances, 
as sketched by Fig.~\ref{f1}. 
\begin{figure}[th]
\centering {\includegraphics[angle=0,scale=0.38]{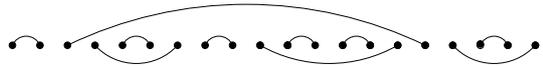}}
\caption{A schematic picture of the RSP. Spin singlets are composed randomly at arbitrary distances.}
\label{f1}
\end{figure}

In order to evaluate ${\cal S}_q$ in the RSP, we begin by considering a number $n$ of singlets connecting   
a contiguous block composed by $L$ spins with the rest of the chain. In this situation, the pseudo-additivity 
of ${\cal S}_q$ implies that Tsallis entropy is given by the Proposition below.   
\begin{proposition}
For a bipartite system composed of a number $n$ of spin-S singlets connecting two blocks, with $n \in \mathbb{N}$, 
Tsallis entropy ${\cal S}_q^{(n)}$ for each block is given by 
\begin{equation}
{\cal S}_q^{(n)} = \frac{1}{1-q} \left[ \left(2S+1)^{n(1-q)} - 1 \right) \right] .
\label{qe-singlets}
\end{equation}
\end{proposition}
\begin{proof}
The proof can be obtained by finite induction. Indeed, the single-site reduced density 
operator $\rho_A$ for a spin-S singlet can be represented by a D-dimensional diagonal 
matrix given by $\rho_A = \textrm{diag}\left(D^{-1},D^{-1},\cdots,D^{-1}\right)$, 
with $D = 2S+1$. Therefore, from Eq.~(\ref{qentropy}), we obtain that 
${\cal S}_q^{(1)} = (1-q)^{-1} (D^{1-q} - 1)$. For two singlets, the pseudo-additivity of ${\cal S}_q$ given 
by Eq.~(\ref{pseudo-add}) implies that ${\cal S}_q^{(2)} = (1-q)^{-1} (D^{2(1-q)} - 1)$. By taking  
the general expression for the entropy for $n$ singlets as ${\cal S}_q^{(n)} = (1-q)^{-1} (D^{n(1-q)} - 1)$, 
we obtain for $(n+1)$ singlets that ${\cal S}_q^{(n+1)} = (1-q)^{-1} (D^{(n+1)(1-q)} - 1)$. Hence, 
Eq.~(\ref{qe-singlets}) holds for any $n \in \mathbb{N}$.
\end{proof}

Tsallis entropy for the RSP can then be obtained by numerically averaging ${\cal S}_q^{(n)}$ over a sample of random 
couplings along the chain. These random configurations are generated by following a gapless probability distribution, which 
drives the system to the RSP, with the entropy of each configuration computed by counting the spin singlets via a renormalization 
group approach described in the next section.

%%%%%%%%%%%%%%%%%%%%%%%%%%%%%%%%%%%%%%%%%%%%%%%%%%%%%%%%%%%%%%%%
\section{Renormalization group method for random spin systems}
%%%%%%%%%%%%%%%%%%%%%%%%%%%%%%%%%%%%%%%%%%%%%%%%%%%%%%%%%%%%%%%%

The RSP can be conveniently handled via a perturbative real-space renormalization group method introduced by Ma, Dasgupta 
and Hu (MDH)~\cite{MDH1,MDH2}, which was successfully applied to the spin-1/2 REHAC. This approach was proven to be 
asymptotically exact, which allowed for a fully characterization of the properties of the RSP~\cite{Fisher:94}. 
Considering a set of random Heisenberg antiferromagnetic interactions $J_i$ between neighbouring spins $S_i$ and $S_{i+1}$, 
the original MDH method consists in finding the strongest interaction $\Omega$ between a pair of spins 
($S_2$ and $S_3$ in Fig. 2a) and treating the couplings of this pair with its neighbors 
($J_1$ and $J_2$ in Fig. 2a) as a perturbation. Diagonalization of the strongest bond leads at zeroth 
order in perturbation theory to a singlet state between the spins coupled by $\Omega$. Then, the singlet is 
decimated away and an effective interaction $J^\prime$ is perturbatively evaluated. By iteratively applying 
this procedure, the low-energy behavior of the ground state will be given by a collection of singlet pairs 
and the structure of the RSP will naturally appear. 

\begin{figure}[th]
\centering {\includegraphics[angle=0,scale=0.4]{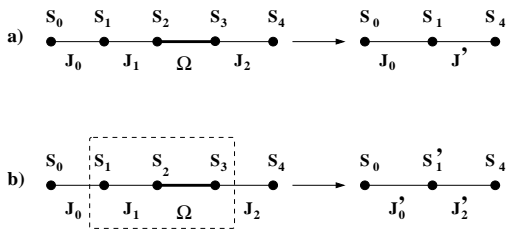}}
\caption{Modified MDH renormalization procedure for spin-$S$ chains.}
\label{f2}
\end{figure}

Unfortunately, when generalized to higher spins, this method, at least in its simplest version, revealed to 
be ineffective. The reason is that, after the elimination procedure of the strongest bond $\Omega$, the effective 
interaction $J^{\prime}$ may be greater than $\Omega$. Then, the problem becomes essentially non-perturbative 
for arbitrary distributions of exchange interactions. For instance, considering the REHAC 
with arbitrary spin-$S$, the renormalized coupling is given by the recursive relation~\cite{Boechat:96}
\begin{equation}
J^{\prime}=\frac{2}{3} S(S+1) \frac{J_1 J_2}{\Omega}.
\label{JPrimeH}
\end{equation}
Notice that, for $S \ge 1$, the renormalization factor is $(2/3)S(S+1)>1$, resulting in the breakdown of 
perturbation theory. In order to solve this problem, a generalization of the MDH method was proposed in 
Refs.~\cite{Saguia:02,Saguia:03} (for other proposals, see also~\cite{Hyman:97,Monthus:97}). This generalized 
MDH method consists in either of the following procedures shown in Fig.~\ref{f2}. Taking the 
case of the Heisenberg chain as an example, if the largest neighboring interaction to $\Omega$, say $J_1$, is 
$J_1 < 3 \Omega /[2S(S+1)]$, then we eliminate the strongest coupled pair obtaining an effective interaction 
between the neighbors to this pair which is given by Eq.~(\ref{JPrimeH}) (see Fig.~2a). This new effective 
interaction is always smaller than those eliminated. Now suppose $J_1 > J_2$ and $J_1> 3 \Omega /[2S(S+1)]$. 
In this case, we consider the {\em trio} of spins-$S$ coupled by the two strongest interactions of the trio, 
$J_1$ and $\Omega$ and solve it exactly (see Fig.~2b). This trio of spins is then substituted by one effective 
spin interacting with its neighbors through new renormalized interactions obtained by degenerate perturbation 
theory for the ground state of the trio. This method has been used to successfully investigate the quantum phase 
diagram of the spin-1~\cite{Saguia:02} and spin-3/2~\cite{Saguia:03} Heisenberg spin chains and will be here 
applied to the computation of the nonadditive entropy.

%%%%%%%%%%%%%%%%%%%%%%%%%%%%%%%%%%%%%%%%%%%%%%%%%%%%%%%%%%%%%%%%
\section{Scaling of nonadditive entropy in the RSP}
%%%%%%%%%%%%%%%%%%%%%%%%%%%%%%%%%%%%%%%%%%%%%%%%%%%%%%%%%%%%%%%%

Let us apply now the generalized MDH approach to analyze the behavior of ${\cal S}_q$ in the RSP of  
antiferromagnetic spin chains. We start with a REHAC, whose Hamiltonian is given by
\begin{equation}
H_{Heis} = \sum_{i=1}^{N} J_i \, \overrightarrow{S}_i \cdot \overrightarrow{S}_{i+1} 
\label{HH}
\end{equation}
where $\{J_i\}$ are random exchange couplings obeying a probability distribution $P(J)$ and 
$\{\overrightarrow{S}_i\}$ are spin-S operators,
%%% v3
with periodic boundary conditions adopted.
%%%
The numerical investigation of ${\cal S}_q$ is 
performed as follows. We begin by considering spin chains with $200,000$ sites, whose couplings $\{J_i\}$ 
are randomly generated by using a gapless power law distribution $P(J) \propto J^{-0.8}$, 
for which trio renormalizations are negligible. Results are then obtained 
by averaging over a sample of $M = 40,000$ configurations for $\{J_i\}$~\footnote{It is worth 
observing that a considerably larger amount of configurations is demanded for the evaluation of 
Tsallis entropy in comparison with the von Neumann entropy (see,  e.g., Ref.~\cite{Saguia:07}).}.
For each random configuration $j$, we decimate the spins out via the generalized MDH technique and compute  
the number $n_j$ of singlets that cross a block of length $L$. The number $n_j$ is counted by 
following Ref.~\cite{Saguia:07}. In turn, if a singlet is
decimated and the spins composing the singlet are in different blocks, this singlet adds one to the 
total number $n_j$. On the other hand, in the case of a trio elimination, nothing is added to $n_j$, 
since one effective spin is returned to the chain (see also Ref.~\cite{Refael:04} for a similar 
approach). The decimation procedure is iterated until the elimination of all spins in the chain. 
After all configurations computed, Tsallis entropy is then obtained by averaging over all 
of them, i.e.,
\begin{equation}
{\cal S}_q = \sum_{j=1}^{M} \frac{{\cal S}_q^{(n_j)}}{M}, 
\end{equation}
with ${\cal S}_q^{(n_j)}$ given by Eq.~(\ref{qe-singlets}). 
The results for spin-1/2, spin-1, and spin-3/2 REHACs are shown in Figs.~\ref{f3}-\ref{f5}. 
\begin{figure}
\centering {\includegraphics[angle=0,scale=0.32]{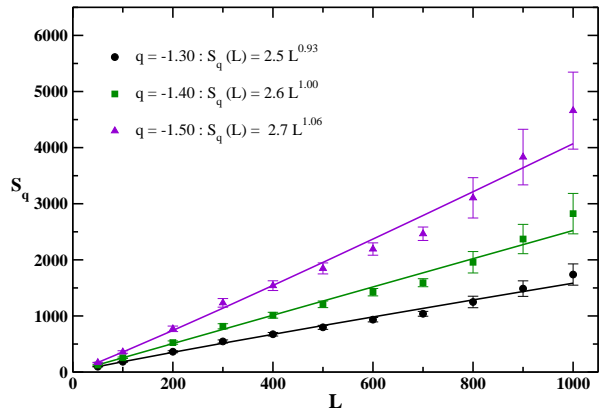}}
\caption{(Color online) Nonadditive entropy ${\cal S}_q(L)$ as a function of the block size $L$ for 
a spin-1/2 REHAC for $q=-1.30$ until $q=-1.50$.}
\label{f3}
\end{figure}
\begin{figure}
\centering {\includegraphics[angle=0,scale=0.32]{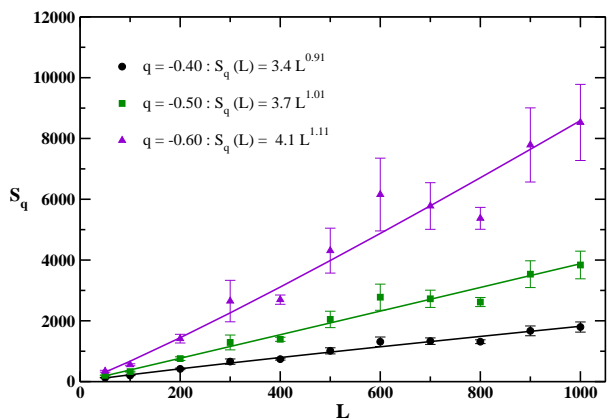}}
\caption{(Color online) Nonadditive entropy ${\cal S}_q(L)$ as a function of the block size $L$ for 
a spin-1 REHAC for $q=-0.40$ until $q=-0.60$.}
\label{f4}
\end{figure}
Note that an extensive ${\cal S}_q$ can be found for specific negative values $q_{ext}$ in the examples 
above. These values are clearly distinct for different spin magnitudes $S$. In turn, a remarkable fact 
is that $q_{ext}$ depends only on $S$ for a system in the RSP and not on the specific model 
considered. In fact, let us consider the spin-1 random exchange biquadratic antiferromagnetic chain, 
whose Hamiltonian is given by 
\begin{equation}
H_{Biq} = \sum_{i=1}^{N} J_i \left(\overrightarrow{S}_i \cdot \overrightarrow{S}_{i+1}\right)^2 .
\label{HB}
\end{equation}
Application of the generalized MDH procedure here results
only in the formation of singlets, with the renormalized exchange
coupling reading~\cite{Boechat:96}
\begin{equation}
J^{\prime}=\frac{2}{9} \frac{J_1 J_2}{\Omega}.
\label{JPrimeB}
\end{equation} 
where $J_1$ and $J_2$ are the nearest neighbors of the strongest bond $\Omega$. 
%%% v3
As in the Heisenberg case, we consider spin chains with $200,000$ sites, whose couplings $\{J_i\}$ 
are randomly generated by using a gapless power law distribution $P(J) \propto J^{-0.8}$ and 
averaged over a sample of $M = 40,000$ configurations for $\{J_i\}$.
%%%%
This model produces 
the same result for $q_{ext}$ as the spin-1 REHAC in the RSP. In order to illustrate 
this equivalence and obtain a numerical value for $q_{ext}$, let us consider the relationship between 
${\cal S}_q$ and $L$ close to the extensivity index $q_{ext}$. 
\begin{figure}
\centering {\includegraphics[angle=0,scale=0.32]{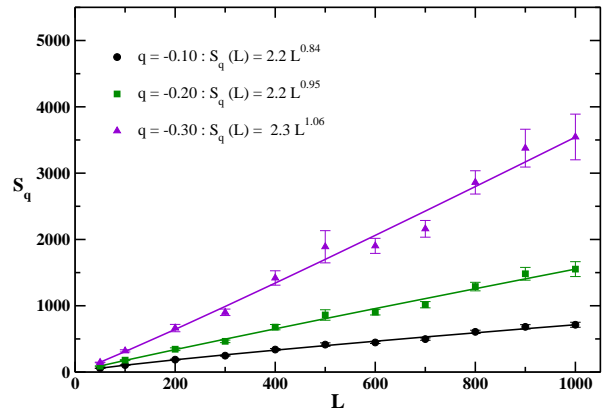}}
\caption{(Color online) Nonadditive entropy ${\cal S}_q(L)$ as a function of the block size $L$ for 
a spin-3/2 REHAC for $q=-0.1$ until $q=-0.3$.}
\label{f5}
\end{figure}
\begin{figure}
\centering {\includegraphics[angle=0,scale=0.32]{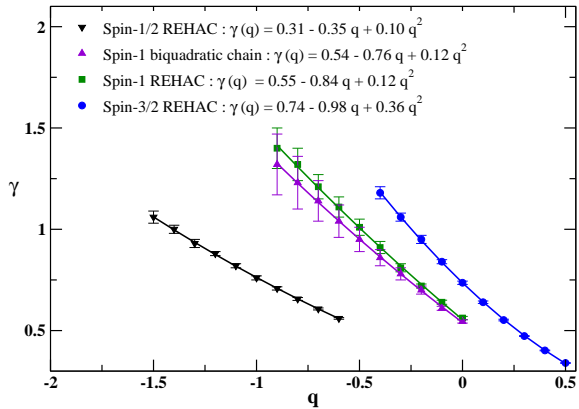}}
\caption{(Color online) Exponent $\gamma$ of the power law ${\cal S}_q (L)\propto L^\gamma$ as a 
function of $q$. We fit $\gamma$ as a quadratic function of $q$, with the extensivity index $q_{ext}$ 
varying only with the magnitude $S$ of spin in the chain. In particular, note the collapse of the 
spin-1 REHAC and biquadratic curves, indicating that $q_{ext}$ is uniquely determined by $S$.}
\label{f6}
\end{figure}
As indicated in Figs.~\ref{f3}-\ref{f5}, we obtain that ${\cal S}_q \propto L^\gamma$
%%%%%%% v2
around the extensive point $\gamma=1$. Indeed, the aim of Figs.~\ref{f3}-\ref{f5} is precisely 
show this dependence as $q$ is varied across extensivity.  
%%%%%%%%%%
Moreover, as shown in Fig.~\ref{f6} for the   
REHACs (with different spin magnitudes $S$) as well as for the spin-1 biquadratic chain, the exponent 
$\gamma$ shows a quadratic dependence on $q$, namely, 
\begin{equation}
\gamma = u q^2 + v q+ w, 
\label{gamma-q}
\end{equation}
with $u,v,w \in \mathbb{R}$. 
Note in Fig.~\ref{f6} the collapse of the spin-1 curves for different models, which indicates that 
$q_{ext}$ is a unique function of the spin of the chain. Therefore, we can associate $q_{ext}$ 
with the effective central charge that governs the entanglement scaling in 
the RSP, which is given by~\cite{Saguia:07}
\begin{equation}
c_{eff} = \ln\left(2S+1\right).
\label{ceff}
\end{equation}
Moreover, $q_{ext}$ can be directly obtained by imposing $\gamma = 1$ in Eq.~(\ref{gamma-q}). 
The values for $q_{ext}$ as well as their relationship with 
the effective central charge are summarized in Table~\ref{t1}. 
Concerning the estimation of the error bar $\Delta q_{ext}$, it can be directly obtained from a standard error propagation procedure. 
Indeed, from $\Delta \gamma = \left| \partial \gamma / \partial q \right| \Delta q$, we obtain 
$\Delta q_{ext} = \Delta\gamma / \left| 2\, u\, q_{ext} + v\right|$.
\begin{table}[ht]
\centering
\begin{tabular}{c c c c}
\hline
Model & $q_{ext}$  & $c_{eff}$   \\ \hline \hline
Spin-$1/2$ REHAC  & \,\,$-1.40 \pm 0.03$\,\, & \,\,$\ln 2$\,\, \\ \hline
Spin-$1$ REHAC    & \,\,$-0.49 \pm 0.04$\,\, & \,\,$\ln 3$ \,\, \\ \hline
Spin-$1$ biquadratic & \,\,$-0.55 \pm 0.06$\,\, & \,\,$\ln 3$ \,\, \\ \hline
Spin-$3/2$ REHAC  & \,\,$-0.25 \pm 0.02$\,\, & \,\,$\ln 4$ \,\, \\ \hline
\end{tabular}
\caption[table1]{Entropic indices $q_{ext}$ that yield the extensivity of ${\cal S}_q$ for different magnitudes $S$ and their 
corresponding effective central charges.}
\label{t1}
\end{table}
Although the analytical relation between $q_{ext}$ and $S$ (and therefore $c_{eff}$) remains elusive at the present stage, it is 
reasonable to suppose that $q_{ext}$ monotonically increases to unity (Boltzmann-Gibbs-von Neumann entropy) 
as $S$ increases to infinity, where the classical limit is obtained. This behavior is indeed displayed in Fig.~\ref{f7} in 
terms of $1/c_{eff}$ for REHACs in the RSP. 
\begin{figure}
\centering {\includegraphics[angle=0,scale=0.3]{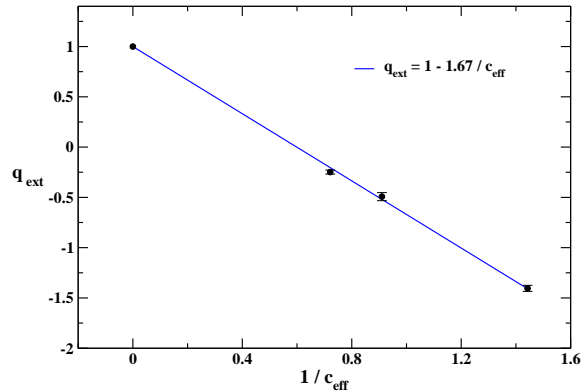}}
\caption{(Color online) Extensivity index $q_{ext}$ as a function of the effective central charge $c_{eff}$ for REHACs in the RSP.}
\label{f7}
\end{figure}
In particular, note that a linear behavior can be inferred between $q_{ext}$ and $1/c_{eff}$, which reads 
\begin{equation}
q_{ext} = 1 - \frac{1.67}{c_{eff}}.
\label{qext-c}
\end{equation}
Therefore, as given by Eq.~(\ref{qext}) for the pure case, we can also determine an expression for $q_{ext}$ in terms of $c_{eff}$ 
for disordered systems, which reinforces the universal properties of the entropic index $q$ in critical systems. 

%%%%%%%%%%%%%%%%%%%%%%%%%%%%%%%%%%%%%%%%%%%%%%%%%%%%%%%%%%%%%%%%
\section{Conclusions}
%%%%%%%%%%%%%%%%%%%%%%%%%%%%%%%%%%%%%%%%%%%%%%%%%%%%%%%%%%%%%%%%

In summary, we have investigated the scaling of Tsallis entropy ${\cal S}_q$ in spin-S random 
critical quantum spin chains. By focusing on the RSP, we have shown that, for specific values 
$q_{ext}$ of the entropic index, ${\cal S}_q$ becomes an extensive quantity, which reconciles 
the quantum scaling with the Clausius-like prescription for classical thermodynamics. 
%%%%%%%%%% v2
It is important to emphasize that the extensivity of the nonadditive entropy does not imply  
absence of correlations between the parts of the quantum system. In this context, there is no 
contradiction with the behavior of block entanglement as measured by the von Neumann entropy.  
%%%%%%%%%%%%
Remarkably, 
$q_{ext}$ is directly associated with the effective central charge $c_{eff}$ that governs 
bipartite entanglement in random spin chains, which means that $q_{ext}$ is solely determined by 
the magnitude $S$ of the spin in the chain.  Moreover, we have inferred a linear algebraic relationship 
between $q_{ext}$ and $1/c_{eff}$. An analytical investigation of this relationship as well as the 
behavior of $q_{ext}$ in other critical phases of random spin chains (e.g., Griffiths phase) are 
challenges under research. 
%%%%%%%%%% v3
Moreover, the scaling of ${\cal S}_q$ far from the extensivity point is also relevant in connection 
with the scaling of ${\textrm{Tr}} \rho^q$ (see, e.g., Refs.~\cite{Cirac:09,Calabrese:10}). We intend to address such 
topics in a future work.

%%%%%%%%%%%%%%%%%%%%%%%%%%%%%%%%%%%%%%%%%%%%%%%%%%%%%%%%%%%%%%%%%%%%
\subsection*{Acknowledgments}
%%%%%%%%%%%%%%%%%%%%%%%%%%%%%%%%%%%%%%%%%%%%%%%%%%%%%%%%%%%%%%%%%%%%

%%%%%%%%%% v2
We thank Prof. C. Tsallis and Dr. B. Coutinho for helpful discussions. 
This work was supported by the Brazilian agencies MCT/CNPq and FAPERJ, being 
performed as part of the Brazilian National Institute for Science and Technology 
of Quantum Information (INCT-IQ).  
%%%%%%%%%%%

\end{document}